\documentclass[onecolumn,aps]{revtex4}
\usepackage{mathrsfs}
\usepackage{amsfonts}
\usepackage{dcolumn}
\usepackage{bm}
\usepackage{epsfig}
\usepackage{epstopdf} 
\usepackage{pgf,fancyhdr}
\usepackage{braket}
\usepackage{amsmath}
\usepackage{amssymb}
\usepackage{amsthm}
\usepackage{graphicx}
\usepackage{hyperref}
\usepackage{booktabs}
\usepackage{xcolor,amsmath}
\hypersetup{
colorlinks=true,
linkcolor=blue,
urlcolor=blue,
citecolor=blue,
pdfborder={0 0 0}
}
\usepackage{calrsfs}
\DeclareMathAlphabet{\pazocal}{OMS}{zplm}{m}{n}
\newtheorem{theorem}{Theorem} 
\newtheorem{corollary}{Corollary}
\theoremstyle{definition} 

\newtheorem{example}{Example}
\theoremstyle{remark}    

\newcommand{\norm}[1]{\lVert#1\rVert}
\begin{document}
\title{Quantum Separability Criteria Based on Symmetric Measurements}
\author{Yu Lu$^{1}$}%
\author{Wen Zhou$^{1}$}
\author{Meng Su$^{1}$}%
\author{Hong-Xing Wu$^{1, 2}$}
\author{Shao-Ming Fei$^{1}$}
\email[]{feishm@cnu.edu.cn}
\author{Zhi-Xi Wang$^{1}$}
\email[]{wangzhx@cnu.edu.cn}
\affiliation{%
$^{1}$School of Mathematical Sciences, Capital Normal University, Beijing 100048, China\\
$^{2}$School of Mathematics and Computational Science, Shangrao Normal University, Shangrao 334001, China}

\begin{abstract}
We propose experimentally feasible separability criteria for bipartite systems based on local symmetric measurements. Through detailed examples, we demonstrate that our criteria can detect
entanglement more effectively compared to existing counterparts. Furthermore,
we demonstrate the potential for our results to be generalized to general multipartite
systems.
\end{abstract}


\maketitle

\section{Introduction}
Quantum entanglement plays significant roles in various quantum tasks such as quantum cryptographic protocols \cite{horodecki2009quantum,nielsen2010quantum,chu2021nonlinear,ekert1991quantum,deutsch1996quantum,hillery1999quantum,long2002theoretically,gao2005deterministic,shao2016entanglement,gao2016dynamics}, quantum teleportation \cite{bennett1993teleporting,gao2008optimal}, dense coding \cite{bennett1992communication} and quantum technologies \cite{xie2016quantum,wang2016experimental,li2016quantum,wang2016quantum,ding2016high,deng2017quantum}. However, the detection of quantum entanglement has always been a great challenge.
Although many important results in the characterization of entanglement for bipartite systems \cite{peres1996separability,chen2002matrix,horodecki1999reduction,guhne2007covariance} and multipartite systems \cite{huber2014witnessing,gao2010detection,gao2011separability,guhne2010separability,huber2010detection,gao2014permutationally,gao2013efficient,hong2016detecting,hong2012measure} have been presented, efficient and universal methods for detecting entanglement remain a central challenge in quantum information science.

There have been many separability criteria such as positive partial transposition (PPT) criterion \cite{peres1996separability,horodecki1996necessary,horodecki1997separability}, realignment criterion \cite{rudolph2003some,horodecki2006separability,chen2002generalized,albeverio2003generalized}, covariance matrix criterion \cite{guhne2007covariance} and correlation matrix criterion \cite{de2008further}. Some classes of quantum entanglement criteria  based on measurements have also been derived. One of them is mutually unbiased bases (MUBs), which are fundamentally connected to the nature of quantum information \cite{wootters1989optimal,wehner2010entropic,barnett2009quantum,durt2010mutually}. In Ref.\cite{spengler2012entanglement}, the authors utilized MUBs to derive separability criteria for bipartite and multipartite systems and continuous-variable quantum systems. For $d\times d$-dimensional bipartite systems, with $d$ a prime power, this criterion is both necessary and sufficient for isotropic states.  The concept of MUBs has been generalized to mutually unbiased measurements (MUMs) \cite{kalev2014mutually}. A complete set of $d+1$ MUMs  can be  constructed in $d$-dimensional Hilbert spaces \cite{kalev2014mutually}, regardless of whether $d$ is a prime power or not. In Ref.\cite{chen2014entanglement}, the authors presented a separability criterion based on MUMs for arbitrary bipartite systems.

Another well-known class of quantum measurements is  symmetric informationally complete positive operator-valued measures (SICPOVMs). In Ref.\cite{shang2018enhanced}, based on SICPOVMs the authors provided a separability criterion which is stronger than those based on local measurements of orthogonal observables. In Ref.\cite{appleby2007symmetric}, the authors  introduced the concept of general SICPOVMs (GSICPOVMs) in which the measurement  operators are not necessarily of rank $1$. In Ref.\cite{gour2014construction}, the authors proposed a method to construct the measurement operators by using generalized Gell-Mann matrices.
Based on GSICPOVMs, many new separability criteria have been derived for both bipartite systems and multipartite systems \cite{liu2017separability,xi2016entanglement,chen2015general,lai2018entanglement}.
The framework of symmetric $(N,M)$ POVMs introduced in Ref.\cite{Siudzinska2022}, which unifies MUMs and GSICPOVMs. Within this framework, the authors derived new entanglement criteria. Additionally, Ref.\cite{wang2025symmetric} used the framework to establish lower bounds on concurrence.

In this paper, we present a separability criterion   by introducing matrices given by the probabilities of local $(N,M)$ POVMs measurement outcomes. The rest of the paper is arranged as follows. In Sec. \ref{sec:2}, we recall some basic notions and properties of $(N,M)$ POVMs. In Sec. \ref{sec:3}, we propose separability criteria based on $(N, M)$ POVMs for bipartite quantum systems. In  Sec. \ref{sec:4}, we generalize the results of bipartite systems to tripartite and multipartite quantum systems.  Some concluding remarks are given in Sec. \ref{sec:5}.

\section{Preliminary}\label{sec:2}
The relevant definition of $(N,M)$ POVMs is elaborated in Ref.\cite{Siudzinska2022}, where an $(N,M)$ POVM is given by $N$ $d$-dimensional POVMs, $\{E_{\alpha,k} \mid k=1,2,\ldots,M\}$ ($\alpha=1,2,\ldots,N$), each with $M$ measurement operators satisfying
\begin{flalign}\nonumber\label{Eq:1}
\mathrm{tr}(E_{\alpha,k}) &= \dfrac{d}{M}, \\ \nonumber
\mathrm{tr}(E_{\alpha,k}^{2}) &= x,\\
\mathrm{tr}(E_{\alpha,k}E_{\alpha,l}) &= \dfrac{d-Mx}{M(M-1)}~~(l\neq k),\\
\mathrm{tr}(E_{\alpha,k}E_{\beta,l}) &= \dfrac{d}{M^{2}}~~(\beta\neq\alpha),\nonumber
\end{flalign}
where the parameter $x$ satisfies the inequality $\dfrac{d}{M^{2}}<x\leq \min\left\{\dfrac{d^{2}}{M^{2}},\dfrac{d}{M}\right\}$.
An $(N,M)$ POVM is called informationally complete if $N(M-1)=d^{2}-1$. Let $\{G_{0}=I_{d}/\sqrt{d},\,G_{\alpha,k}|~\alpha=1,\ldots,N;\,k=1,\ldots,M-1\}$
be an orthonormal Hermitian operator basis such that $\mathrm{tr}(G_{\alpha,k})=0$, where $I_{d}$ denotes the $d\times d$ identity operator. The measurement operators of the informationally complete $(N,M)$ POVM
have the following explicit expressions:
\begin{equation}\label{Eq:2}
E_{\alpha,k}=\dfrac{1}{M}I_{d}+tH_{\alpha,k},
\end{equation}
where
\begin{flalign}\label{Eq:3}
 	H_{\alpha,k}=\begin{cases}
 		G_{\alpha}-\sqrt{M}(\sqrt{M}+1)G_{\alpha,k}~~&(k=1,\cdots,M-1),\\[1mm]
 		(\sqrt{M}+1)G_{\alpha}~~&(k=M),
 	\end{cases}
 \end{flalign}
with $G_{\alpha}=\sum\limits_{k=1}^{M-1}G_{\alpha,k}$. The positivity of $E_{\alpha,k}$ implies that
 \begin{equation*}
 	-\dfrac{1}{M}\dfrac{1}{\lambda_{\max}}\leq t \leq \dfrac{1}{M}\dfrac{1}{|\lambda_{\min}|},
 \end{equation*}
where $\lambda_{\text{max}}$ and $\lambda_{\text{min}}$ are the minimal and maximal eigenvalues, respectively, among the eigenvalues of all $H_{a,k}$. The parameters $t$ and $x$ subject to the following relation, then
 \begin{equation}\label{Eq:4}
 	x=\dfrac{d}{M^{2}}+t^{2}(M-1)(\sqrt{M}+1)^{2}.
 \end{equation}

{\it GSICPOVMs}~~ If $N = 1$ and $M = d^{2}$, the $(N,M)$ POVMs reduces                                                                                                                                                                                                           to GSICPOVMs \cite{appleby2007symmetric, gour2014construction}, given by $d^2$ positive semidefinite operators $\{M_k\}_{\alpha =1}^{d^2}$ satisfying
\begin{flalign}\label{GSIC}
 	\mathrm{tr}(M_{k}) &= \dfrac{1}{d}, \nonumber \\
 	\mathrm{tr}(M_{k}^{2}) &= a,\\
 	\mathrm{tr}(M_{k}M_{l}) &= \frac{1-da}{d(d^2-1)}~~(l\neq k),\nonumber
\end{flalign}
where the parameter $a$ satisfies the relation $\frac{1}{d^3} <a \leq \frac{1}{d^2}$. When $a={1}/{d^2}$, $M_\alpha$ are rank $1$ projectors and the GSICPOVMs become SICPOVMs.
Let $\{F_k\}_{k=1}^{d^2-1}$ be a set of $d^2-1$ Hermitian and traceless operators on the complex vector space $\mathbb{C}^d$, one has \cite{gour2014construction}
\begin{equation}\label{Eq:6}
M_{k}=\dfrac{1}{M}I_{d}+t\mathbb{F}_{k},
\end{equation}
where
\begin{flalign}\label{Eq:7}
 	\mathbb{F}_{k}=\begin{cases}
 		F-d(d+1)F_k~~&(k=1,\cdots,d^2-1),\\[1mm]
 		(d+1)F~~&(k=d^2),
 	\end{cases}
 \end{flalign}
with $F=\sum\limits_{k=1}^{d^2-1}F_{k}$. To ensure  $M_{k}\geq 0$ and $t$ must satisfy the following relation,
 \begin{equation}\label{Eq:8}
 	-\dfrac{1}{d^2}\dfrac{1}{\lambda_{\max}}\leq t \leq \dfrac{1}{d^2}\dfrac{1}{|\lambda_{\min}|},
 \end{equation}
where $\lambda_{\rm max}$ and $\lambda_{\rm min}$ are the minimal and maximal eigenvalues among the eigenvalues  of $F_{k}$ for all $k$. Additionally, the parameter $a$ in Eq. (\ref{GSIC}) is given by
\begin{align}\label{Eq:9}
a=\frac{1}{d^3}+t^2(d-1)(d+1)^3.
\end{align}

{\it MUMs}~~If $N=d+1$ and $M=d$,  the $(N,M)$ POVMs reduce to the MUMs \cite{kalev2014mutually}.
A set of POVMs $\{\pazocal{P}^{(b)}\}_{b=1}^{d+1}$ on $\mathbb{C}^d$, where $\pazocal{P}^{(b)}=\{P_{n}^{(b)}\}_{n=1}^{d}$, is called MUMs if
\begin{equation}\label{MUM}
\begin{split}
\mathrm{tr}(P_{n}^{(b)})&=1,\\
\mathrm{tr}(P_{n}^{(b)})^2&=\kappa,\\
\mathrm{tr}(P_{n}^{(b)}P_{n'}^{(b)})&=\frac{1-\kappa}{d-1} ~~~(n\neq n'),\\
\mathrm{tr}(P_{n}^{(b)}P_{n'}^{(b')})&=\frac{1}{d}~~~(b\neq b'),
\end{split}
\end{equation}
where the parameter $\kappa$ satisfies the inequality $\frac{1}{d} < \kappa \leq 1$. Let $\{F_{n,b}:n=1,2,\cdots,d-1; b=1,2,\cdots,d+1\}$ be
$d^2 - 1$ Hermitian and traceless operators such that $\mathrm{tr}(F_{n,b}F_{n',b'})=\delta_{n,n'}\delta_{b,b'}$. MUMs have the following expressions:
\begin{equation}\label{eq2}
P_{n}^{(b)}=\frac{1}{d}I+tF_{n}^{(b)},
\end{equation}
where
\begin{equation}\label{eq1}
F_{n}^{(b)}=
\begin{cases}
   F^{(b)}-(d+\sqrt{d})F_{n,b}~~&(n=1,2,\cdots,d-1),\\[2mm]
   (1+\sqrt{d})F^{(b)}~~&(n=d),
\end{cases}
\end{equation}
with $F^{(b)}=\sum_{n=1}^{d-1}F_{n,b}$, $b=1,2,\ldots,d+1$. To ensure  $P_{n}^{(b)} \geq 0$ and the parameter $t$ must satisfy the inequality
\begin{equation}\label{Eq:13}
-\frac1{d}\frac1{\lambda_{\rm max}}\leq t\leq\frac1{d}\frac1{|\lambda_{\rm min}|},
\end{equation}
where $\lambda_{\rm min} = \min_b\lambda_{\rm min}^{(b)}$ and $\lambda_{\rm max} = \max_b\lambda_{\rm max}^{(b)}$, with $\lambda_{\rm min}^{(b)}$ ($\lambda_{\rm max}^{(b)}$) being the largest positive (smallest negative) eigenvalues of $F_n^{(b)}$ for all $n = 1,\ldots,d$ given in Eq.(\ref{eq1}).
To construct MUMs in Eq.(\ref{eq2}), the parameter $\kappa$ in Eq.(\ref{MUM}) is given by
\begin{equation}\label{Eq:14}
\kappa=\frac{1}{d}+t^{2}(1+\sqrt{d})^{2}(d-1).
\end{equation}

Subsequently, building upon these new concepts, a series of efficient entanglement criteria based on MUMs were derived in Refs.\cite{chen2014entanglement,liu2017separability,shen2015entanglement,shen2018improved}. More elegant criteria have also been obtained in terms of SICPOVMs, GSICPOVMs and $(N,M)$ POVMs, as summarized in Table \ref{tab:1}.

\begin{table}[h!]
\centering
\caption{Some related separability criteria}
\begin{tabular}{cccc}
\hline
 &  MUMs & GSICPOVMs & $(N,M)$ POVMs \\
\hline
Based on the trace \\of measurement operators & \cite{chen2014entanglement} \quad& \cite{chen2015general} \quad&\cite{Siudzinska2022} \\
\hspace{1mm}\\
Based on measurements \\on $(\rho - \rho^A\otimes\rho^B)$ & \cite{shen2015entanglement} \quad& \cite{shen2015entanglement} \quad&\cite{tang2023improved}  \\
\hspace{1mm}\\
Multipartite and bipartite systems \\with different individual dimensions & \cite{liu2015separability} \quad& \cite{xi2016entanglement} \quad&\cite{tang2023enhancing} \\
\hspace{1mm}\\
For dimensions \(d'=sd+r\) & \cite{liu2017separability} \quad& \cite{liu2017separability} \quad&\cite{tang2023improved} \\
\hspace{1mm}\\
Criteria based on trace norm &\cite{shen2018improved} \quad& \cite{lai2018entanglement} \quad&\cite{Siudzinska2022} \\
\hline
\end{tabular}
\label{tab:1}
\end{table}

\section{Detecting entanglement of bipartite states}\label{sec:3}

Let $H_X$ be a $d_X$-dimensional Hilbert space associated with system $X$. A bipartite state $\rho\in H_{A}\otimes H_{B}$ is said to be separable if it can be represented as a convex sum of tensor products of the local states of the subsystems, which is
\begin{equation}\label{sep}
\rho_{AB}=\sum_{i}p_i\rho^{i}_{A}\otimes \rho^{i}_{B},
\end{equation}
where $p_i\geq 0$ and $\sum_i p_i=1$. Otherwise $\rho$ is said to be entangled.

Consider local $(N,M)$ POVMs on a bipartite state $\rho\in H_{A}\otimes H_{B}$. With respect to the measurement operators $E_{\alpha, k} \otimes E_{\beta, j}$,  we denote the matrix $\mathcal{P}(\rho)$ with entries given by the probabilities $\mathrm{tr}[(E_{\alpha, k} \otimes E_{\beta, j})\rho]$, where $\alpha, k$ and $\beta, j$ are the row and the column indices, respectively. Define
\begin{align}\label{eq:3}
\mathcal{M}_{\mu,\nu}^{l}(\rho_{AB})=\begin{pmatrix}
	\mu\nu J_{l\times l}& \mu \omega_{l}(\sigma)^T\\[1mm]
	\nu \omega_{l}(\tau)& \mathcal{P}(\rho_{AB})
\end{pmatrix},
\end{align}
where
\begin{align}\label{Eq:17}
\tau = \left( \begin{array}{c}
\mathrm{tr}(E_{1,1} \rho_A) \\
\mathrm{tr}(E_{1,2} \rho_A) \\
\vdots \\
\mathrm{tr}(E_{1,M_A} \rho_A) \\
\mathrm{tr}(E_{2,1} \rho_A) \\
\mathrm{tr}(E_{2,2} \rho_A) \\
\vdots \\
\mathrm{tr}(E_{2,M_A} \rho_A) \\
\vdots \\
\mathrm{tr}(E_{N_A,M_A} \rho_A)
\end{array} \right), \qquad
\sigma = \left( \begin{array}{c}
\mathrm{tr}(E_{1,1} \rho_B) \\
\mathrm{tr}(E_{1,2} \rho_B) \\
\vdots \\
\mathrm{tr}(E_{1,M_B} \rho_B) \\
\mathrm{tr}(E_{2,1} \rho_B) \\
\mathrm{tr}(E_{2,2} \rho_B) \\
\vdots \\
\mathrm{tr}(E_{2,M_B} \rho_B) \\
\vdots \\
\mathrm{tr}(E_{N_B,M_B} \rho_B)
\end{array} \right),
\end{align}
with $\rho_{A}$ ($\rho_{B}$) being the reduced density matrix obtained by tracing over the subsystem $H_A$ ($H_B$) of $\rho_{AB}$. $\mu$ and $\nu$ are real numbers, $l$ is a positive integer, $J_{l \times l}$ is the matrix with all $l \times l$ entries being $1$, $\omega_{l}(X)=\underbrace{(X, \ldots, X)}_{l}$, and $T$ denotes the transpose of a matrix. Thus, we have the following separability criterion:

\begin{theorem}\label{th:1}
If a bipartite state $\rho_{AB}\in H_{A}\otimes H_{B}$ is separable, then
\begin{eqnarray}\label{eq:5}
\left\|\mathcal{M}_{\mu, \nu}^l \left(\rho_{AB}\right)\right\|_{\mathrm{tr}} \leq \sqrt{\left(l \mu^{2}+\frac{(d_{A}-1)(d_{A}^2+M_{A}^2x_A)}{d_{A}M_{A}(M_{A}-1)}\right)\left(l \nu^{2}+\frac{(d_{B}-1)(d_{B}^2+M_{B}^2x_B)}{d_{B}M_{B}(M_{B}-1)}\right)},
\end{eqnarray}
where $x_A$ ($x_B$) is the parameter in the $(N,M)$ POVMs (\ref{Eq:1}) on subsystem $H_A$ ($H_B$); $\|G\|_{\mathrm{tr}}=\mathrm{tr}\left(\sqrt{G^{\dagger} G}\right)$ is the trace norm of a matrix $G$.
\end{theorem}

\begin{proof}
If $\rho_{AB}$ is a separable pure state, $\rho_{AB}=\rho_{A} \otimes \rho_{B}$, we have  $\mathcal{P}(\rho_{AB}) = \tau \sigma^{T}$. It has been shown in Ref.\cite{Siudzinska2022} that the following inequality holds for an arbitrary state $\varrho_{AB}$:
\begin{equation}\label{re:1}
\sum_{\alpha = 1}^{N} \sum_{k = 1}^{M} p_{\alpha,k}^2 \leq \dfrac{(d-1)(xM^{2}+d^{2})}{dM(M-1)},
\end{equation}
where $p_{\alpha,k} = \mathrm{tr}(\varrho_{AB} E_{\alpha,k})$. Therefore, we have
\begin{align}\label{Eq:20}
\|M_{\mu,\nu}^l\left(\rho_{AB}\right)\|_{\mathrm{tr}}
&=\left\|\left(\begin{array}{cc}\mu\nu J_{l\times l} & \mu \omega_{l}(\sigma)^T \\ \nu \omega_{l}(\tau)  & \tau\sigma ^\mathrm{T}\end{array}\right)\right\|_{\mathrm{tr}} \nonumber\\
&=\left\|\left(\begin{array}{cc}\mu J_{l\times 1} & \nonumber\\ \tau & \end{array}\right)\left(\begin{array}{cc}\nu J_{1\times l} & \sigma^\mathrm{T}\end{array}\right)\right\|_{\mathrm{tr}} \nonumber\\
&\leq\left\|\left(\begin{array}{cc}\mu J_{l\times 1} & \nonumber\\ \tau & \end{array}\right)\right\|_{\mathrm{tr}}\left\|\left(\begin{array}{cc}\nu J_{1\times l} & \sigma^\mathrm{T}\end{array}\right)\right\|_{\mathrm{tr}} \nonumber\\
&=\left\|\left(\begin{array}{cc}\mu J_{l\times 1} & \nonumber\\ \tau & \end{array}\right)\right\|_{\mathrm{tr}}\left\|\left(\begin{array}{cc}\nu J_{1\times l} & \sigma^\mathrm{T}\end{array}\right)\right\|_{\mathrm{tr}} \nonumber\\
&\leq \sqrt{\left(l \mu^{2}+\frac{(d_{A}-1)(d_{A}^2+M_{A}^2x_A)}{d_{A}M_{A}(M_{A}-1)}\right)\left(l \nu^{2}+\frac{(d_{B}-1)(d_{B}^2+M_{B}^2x_B)}{d_{B}M_{B}(M_{B}-1)}\right)},
\end{align}
where the first inequality is due to $\parallel A  B \parallel_{\mathrm{tr}} \leq \parallel A \parallel_{\mathrm{tr}} \parallel B \parallel_{\mathrm{tr}}$ for any matrices $A$ and $B$, and the last inequality is from (\ref{re:1}).
\end{proof}

{\it Remark} When both parameters $\mu$ and $\nu$ are zero, $\mu =\nu=0$, Theorem \ref{th:1} reduces to the conclusion presented in Ref.\cite{tang2023enhancing}, namely,
\begin{eqnarray}\label{Eq:21}
\|\mathcal{P}(\rho_{AB})\|_{\mathrm{tr}}\leq  \sqrt{\left(\frac{(d_{A}-1)(d_{A}^2+M_{A}^2x_A)}{d_{A}M_{A}(M_{A}-1)}\right)
\left(\frac{(d_{B}-1)(d_{B}^2+M_{B}^2x_B)}{d_{B}M_{B}(M_{B}-1)}\right)}
\end{eqnarray}
if $\rho_{AB}$ is separable, where $\|\mathcal{P}\|_{\mathrm{tr}}=\mathrm{tr} \sqrt{\mathcal{P}^{\dagger}\mathcal{P}}$.

In particular, as special $(N,M)$ POVMs, the GSICPOVMs and MUMs give rise to the following corollaries correspondingly:

\begin{corollary}\label{co:1}
Consider two GSICPOVMs $\{M_{\alpha}\}_{\alpha=1}^{d_A^{2}}$ with parameter $a_A$ and $\{M_{\beta}\}_{\beta=1}^{d_B^{2}}$ with parameter $a_B$.
If a state $\rho_{AB}\in H_{A}\otimes H_{B}$ is separable, then
\begin{eqnarray}\label{eq:6}
\left\|\mathcal{M}_{\mu, \nu}^l \left(\rho_{AB}\right)\right\|_{\mathrm{tr}} \leq \sqrt{\left(l \mu^{2}+\frac{a_{A}d_{A}^{2}+1}{d_{A}(d_{A}+1)}\right)\left(l \nu^{2}+\frac{a_{B}d_{B}^{2}+1}{d_{B}(d_{B}+1)}\right)},
\end{eqnarray}
where \begin{align}\label{eq:7}
\mathcal{M}^l_{\mu,\nu}(\rho_{AB})=\begin{pmatrix}
	\mu\nu J_{l\times l}& \mu \omega_{l}(\zeta)^T\\[1mm]
	\nu \omega_{l}(\varsigma)& \mathcal{G}(\rho_{AB})
\end{pmatrix},
\end{align}
with
\begin{eqnarray}\label{Eq:24}
\varsigma = \left( \begin{array}{c}
\mathrm{tr}(M_1^A \rho_A) \\
\mathrm{tr}(M_2^A \rho_A) \\
\vdots \\
\mathrm{tr}(M_{d^2_A}^A \rho_A)
\end{array} \right), \qquad
\zeta = \left( \begin{array}{c}
\mathrm{tr}(M_1^B \rho_B) \\
\mathrm{tr}(M_2^B \rho_B) \\
\vdots \\
\mathrm{tr}(M_{d^2_B}^B \rho_B)
\end{array} \right),
\end{eqnarray}
$\mathcal{G}(\rho_{AB})$ is a matrix with entries given by $\mathrm{tr}(M_{\alpha}\otimes M_{\beta}\rho_{AB})$.
\end{corollary}

\begin{corollary}\label{co:2}
Let $\{P^{(b)}\}_{b=1}^{d_A +1}$ with parameter $\kappa_A$ and $\{Q^{(b')}\}_{b'=1}^{d_B +1}$ with parameter $\kappa_B$ being MUMs on subsystems $H_A$ and $H_B$, respectively, where $P^{(b)} = \{P_n^{(b)}\}_{n=1}^{d_A}$, $Q^{(b')} = \{Q_{n'}^{(b')}\}_{n'=1}^{d_B}$, and $b,b' = 1,2,\ldots,d + 1$. Denote $\mathcal{J}(\rho_{AB})$ the matrix with entries given by the probabilities $\mathrm{tr}[(P_n^{(b)} \otimes Q_{n'}^{(b')})\rho_{AB}]$, where $n, b$ and $n', b'$ are the row and the column indices, respectively. For any separable state $\rho_{AB}$ in $H_{A} \otimes H_B$, we have
\begin{footnotesize}
\begin{eqnarray}\label{eq:8}
\left\|\mathcal{M}_{\mu, \nu}^l \left(\rho_{AB}\right)\right\|_{\mathrm{tr}} \leq \sqrt{\left(l \mu^{2} + 1 + \kappa_A\right)\left(l \nu^{2}+ 1 + \kappa_B\right)},
\end{eqnarray}
\end{footnotesize}
where \begin{align}\label{eq:9}
\mathcal{M}_{\mu,\nu}^l(\rho_{AB})=\begin{pmatrix}
	\mu\nu J_{l\times l}& \mu \omega_{l}(\eta)^T\\[1mm]
	\nu \omega_{l}(\theta)& \mathcal{J}(\rho_{AB})
\end{pmatrix},
\end{align}
\begin{eqnarray}\label{Eq:27}
\theta = \left( \begin{array}{c}
\mathrm{tr}(P_1^{(1)} \rho_A) \\
\mathrm{tr}(P_2^{(1)} \rho_A) \\
\vdots \\
\mathrm{tr}(P_{n}^{(1)} \rho_A)\\
\mathrm{tr}(P_1^{(2)} \rho_A) \\
\mathrm{tr}(P_2^{(2)} \rho_A) \\
\vdots \\
\mathrm{tr}(P_{n}^{(2)} \rho_A)\\
\mathrm{tr}(P_{n}^{(3)} \rho_A)\\
\vdots \\
\mathrm{tr}(P_{n}^{(b)} \rho_A)
\end{array} \right), \qquad
\eta = \left( \begin{array}{c}
\mathrm{tr}(Q_1^{(1)} \rho_A) \\
\mathrm{tr}(Q_2^{(1)} \rho_A) \\
\vdots \\
\mathrm{tr}(Q_{n'}^{(1)} \rho_B)\\
\mathrm{tr}(Q_1^{(2)} \rho_B) \\
\mathrm{tr}(Q_2^{(2)} \rho_B) \\
\vdots \\
\mathrm{tr}(Q_{n'}^{(2)} \rho_B)\\
\mathrm{tr}(Q_{n'}^{(3)} \rho_B)\\
\vdots \\
\mathrm{tr}(Q_{n'}^{(b')} \rho_B)
\end{array} \right).
\end{eqnarray}
\end{corollary}

Next, we illustrate  our criteria by several examples.

\begin{example}\label{Ex:1}
We consider the following state by mixing $\rho_{AB}$ with white noise:
\begin{equation}\label{Eq:28}
\rho_{p}=\dfrac{1-p}{9}I_{9}+p\rho_{AB},
\end{equation}
where
\begin{equation}\label{Eq:29}
\rho_{AB} = \dfrac{1}{4}\left(I_9-\sum\limits_{i=0}^{4}\ket{\psi_{i}}\bra{\psi_{i}}\right)
\end{equation}
is a $3\times 3$ PPT entangled state with
    \begin{align*}
    &\ket{\psi_{0}}=\dfrac{\ket{0}(\ket{0}-\ket{1})}{\sqrt{2}}, ~~ \ket{\psi_{1}}=\dfrac{(\ket{0}-\ket{1})\ket{2}}{\sqrt{2}}, ~~ \ket{\psi_{2}}=\dfrac{\ket{2}(\ket{1}-\ket{2})}{\sqrt{2}},\\  &\ket{\psi_{3}}=\dfrac{(\ket{1}-\ket{2})\ket{0}}{\sqrt{2}},~~ \ket{\psi_{4}}=\dfrac{(\ket{0}+\ket{1}+\ket{3})(\ket{0}+\ket{1}+\ket{3})}{3}.
    \end{align*}

Let the $(N,M)$ POVMs in Theorem \ref{th:1} be an $(8,2)$ POVMs, in which the Hermitian basis operator $G_{\alpha, k}$ is defined by the Gell-Mann matrices(see \cite{gell1962symmetries} for details):
\begin{align*}
  &G_{11}=\dfrac{1}{\sqrt{2}}\begin{pmatrix}
    	0 & 1 & 0\\
    	1 & 0 & 0\\
    	0 & 0 & 0
    \end{pmatrix},~
    G_{21}=\dfrac{1}{\sqrt{2}}\begin{pmatrix}
    	0 & -\mathrm{i} & 0\\
    	\mathrm{i} & 0 & 0\\
    	0 & 0 & 0
    \end{pmatrix},~
    G_{31}=\dfrac{1}{\sqrt{2}}\begin{pmatrix}
    	0 & 0 & 1\\
    	0 & 0 & 0\\
    	1 & 0 & 0
    \end{pmatrix},~
    G_{41}=\dfrac{1}{\sqrt{2}}\begin{pmatrix}
    	0 & 0 & -\mathrm{i}\\
    	0 & 0 & 0\\
    	\mathrm{i} & 0 & 0
    \end{pmatrix}, \\[1mm]
  & G_{51}=\dfrac{1}{\sqrt{2}}\begin{pmatrix}
    	0 & 0 & 0\\
    	0 & 0 & 1\\
    	0 & 1 & 0
    \end{pmatrix},~
    G_{61}=\dfrac{1}{\sqrt{2}}\begin{pmatrix}
    	0 & 0 & 0\\
    	0 & 0 & -\mathrm{i}\\
    	0 & \mathrm{i} & 0
    \end{pmatrix},~
  G_{71}=\dfrac{1}{\sqrt{2}}\begin{pmatrix}
    	1 & 0 & 0\\
    	0 & -1 & 0\\
    	0 & 0 & 0
    \end{pmatrix},~
    G_{81}=\dfrac{1}{\sqrt{6}}\begin{pmatrix}
    	1 & 0 & 0\\
    	0 & 1 & 0\\
    	0 & 0 & -2
    \end{pmatrix}.
\end{align*}
Here, the parameter is $x=\frac{3}{4}+t^{2}(\sqrt{2}+1)^{2}$ with $t\in[-0.2536,0.2536]$.
In particular, for $\mu = 0.1$, $\upsilon =0.05$, $l = 2$ and $t = 0.01$, using Theorem \ref{th:1} we obtain
\begin{align}\label{Eq:30}
f_1(p) \nonumber &\equiv \left\|\mathcal{M}_{\mu, \nu}^l \left(\rho_{p}\right)\right\|_{\mathrm{tr}} - \sqrt{\left(l \mu^{2}+\frac{(d_{A}-1)(d_{A}^2+M_{A}^2x_A)}{d_{A}M_{A}(M_{A}-1)}\right)\left(l \nu^{2}+\frac{(d_{B}-1)(d_{B}^2+M_{B}^2x_B)}{d_{B}M_{B}(M_{B}-1)}\right)}\\  &\approx 0.001174p-0.0007841,
\end{align}
which is  positive(i.e., $\rho_p$ is entangled) when $0.670 093 \leq p \leq 1$.
 From Theorem $1$ in Ref.\cite{tang2023enhancing}, we know that
\begin{align}\label{Eq:31}
g_1(p) \nonumber&\equiv \|\mathcal{P}(\rho_p)\|_{\mathrm{tr}}-  \sqrt{\left(\frac{(d_{A}-1)(d_{A}^2+M_{A}^2x_A)}{d_{A}M_{A}(M_{A}-1)}\right)
\left(\frac{(d_{B}-1)(d_{B}^2+M_{B}^2x_B)}{d_{B}M_{B}(M_{B}-1)}\right)} \\& \approx 0.0008892p - 0.0007889
\end{align}
is positive; then $\rho_p$ is entangled, which implies $0.882 182 \leq p \leq 1$. Clearly, our result outperforms that of Ref.\cite{tang2023enhancing}; see Fig.\ref{fig:6}.

When $N = 1$ and $M = 9$,  Theorem \ref{th:1} reduces to Corollary \ref{co:1}. The Hermitian basis operator $G_{\alpha, k}$ is given  by the following Gell-Mann matrices(see \cite{gour2014construction} for details):
\begin{align*}
&G_{11}=\left(
\begin{array}{ccc}
 \frac{1}{\sqrt{2}} & 0 & 0 \\
 0 & -\frac{1}{\sqrt{2}} & 0 \\
 0 & 0 & 0 \\
\end{array}
\right),~~
G_{12}=\left(
\begin{array}{ccc}
 0 & \frac{1}{\sqrt{2}} & 0 \\
 \frac{1}{\sqrt{2}} & 0 & 0 \\
 0 & 0 & 0 \\
\end{array}
\right),~~
G_{13}=\left(
\begin{array}{ccc}
 0 & 0 & \frac{1}{\sqrt{2}} \\
 0 & 0 & 0 \\
 \frac{1}{\sqrt{2}} & 0 & 0 \\
\end{array}
\right),
G_{14}=\left(
\begin{array}{ccc}
 0 & -\frac{i}{\sqrt{2}} & 0 \\
 \frac{i}{\sqrt{2}} & 0 & 0 \\
 0 & 0 & 0 \\
\end{array}
\right),~~\\
&G_{15}=\left(
\begin{array}{ccc}
 \frac{1}{\sqrt{6}} & 0 & 0 \\
 0 & \frac{1}{\sqrt{6}} & 0 \\
 0 & 0 & -\sqrt{\frac{2}{3}} \\
\end{array}
\right),~~
G_{16}=\left(
\begin{array}{ccc}
 0 & 0 & 0 \\
 0 & 0 & \frac{1}{\sqrt{2}} \\
 0 & \frac{1}{\sqrt{2}} & 0 \\
\end{array}
\right),
G_{17}=\left(
\begin{array}{ccc}
 0 & 0 & -\frac{i}{\sqrt{2}} \\
 0 & 0 & 0 \\
 \frac{i}{\sqrt{2}} & 0 & 0 \\
\end{array}
\right),~~
G_{18}=\left(
\begin{array}{ccc}
 0 & 0 & 0 \\
 0 & 0 & -\frac{i}{\sqrt{2}} \\
 0 & \frac{i}{\sqrt{2}} & 0 \\
\end{array}
\right).~~
\end{align*}
Here, the parameter is $a=\frac{1}{27}+128t^{2}$ with $t\in[-0. 012, 0. 012]$.
 Theorem $1$ in Ref.\cite{shen2018improved} detects the entanglement of $\rho_p$ within the interval $0.8822 \leq p \leq 1$. Figure.\ref{fig:mesh} shows the variation of the entanglement detection by our Theorem $1$ with respect to parameters $\mu$ and $\nu$, with fixed $p = 0.8822$ and $l = 1$.  The blue region is the zero plane, and the orange region corresponds to the function of $f(\mu, \nu) = \left\|\mathcal{M}_{\mu, \nu}^l \left(\rho_p\right)\right\|_{\mathrm{tr}} -\sqrt{\left(l \mu^{2}+\frac{a_{A}d_{A}^{2}+1}{d_{A}(d_{A}+1)}\right)\left(l \nu^{2}+\frac{a_{B}d_{B}^{2}+1}{d_{B}(d_{B}+1)}\right)}$. It is observed that there are more parameter combinations of $\mu$ and $\nu$ that can be used to detect the entanglement of $\rho_p$ within the range $0.8822 \leq p \leq 1$. The above-mentioned numerical computations show that the existing criteria \cite{spengler2012entanglement,chen2014entanglement,shen2015entanglement} based on MUMs cannot detect any entanglement of $\rho_p$.
\begin{figure}
\includegraphics[width=0.80\textwidth]{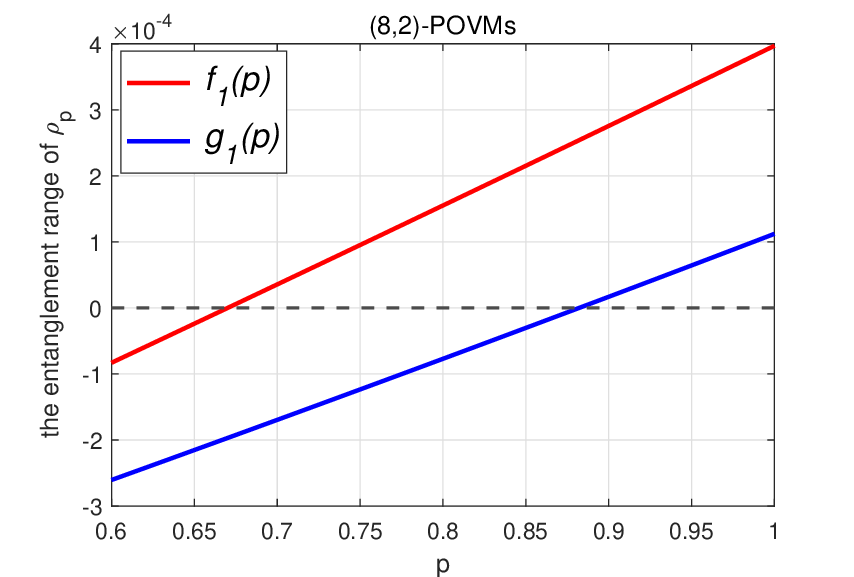}
        \label{fig:6}
\caption{ ~\(f_1(p)\)~ derived from Theorem \ref{th:1} (solid red line),from which~ \(\rho_p\)~ is entangled for~  \(0.670 093 \leq p \leq 1\).   \(g_1(p)\)~ presented in Theorem $1$ of Ref.\cite{tang2023enhancing} (dashed blue line), from which~ \(\rho_p\)~ is entangled for~ \(0.882 182 \leq p \leq 1\).}
\end{figure}

\begin{figure}
 \includegraphics[width=0.80\textwidth]{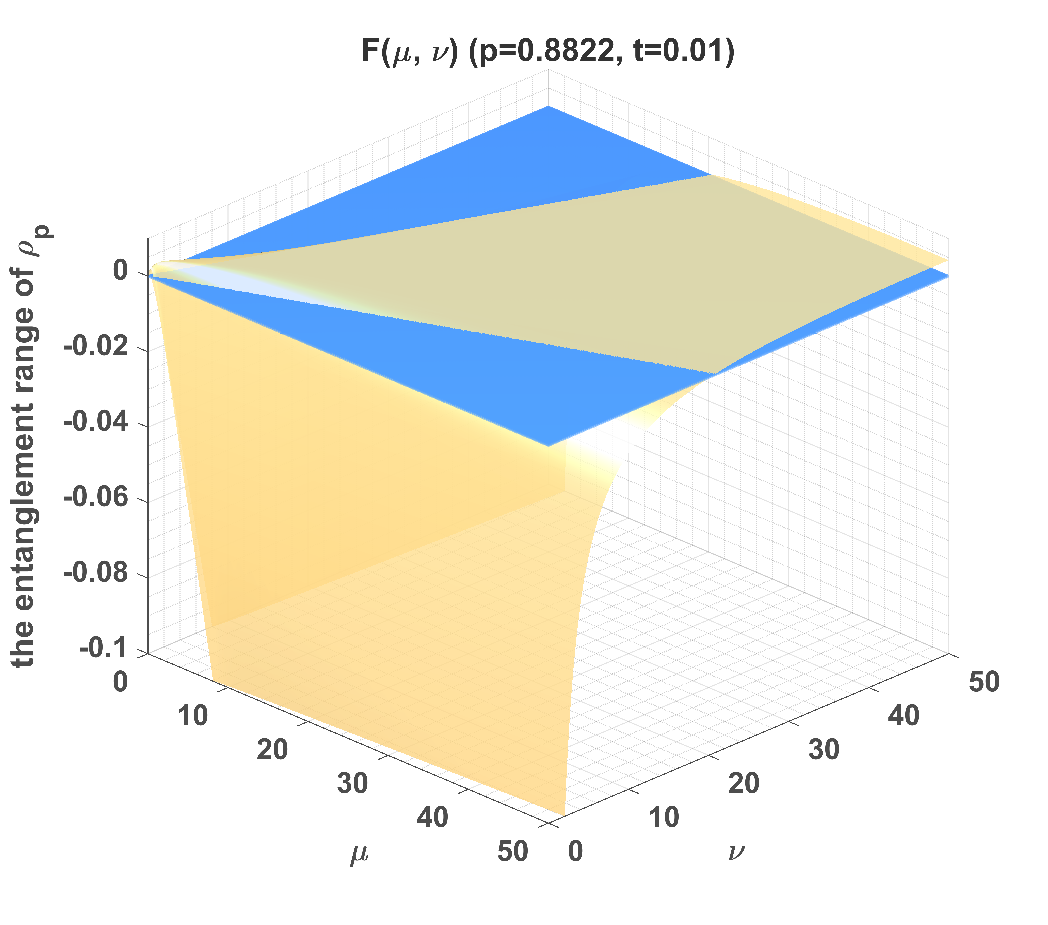}
    \caption{The orange region above the blue zero plane indicates that entanglement of ~$\rho_p$~ for ~$0.8822 \leq p \leq 1$~ is detected by  Theorem $1$ in Ref.\cite{shen2018improved}.}
    \label{fig:mesh}
\end{figure}

In particular, for $\mu = 0.1$, $\upsilon =0.05$, $l = 2$ and $t = 0.01$, from Corollary \ref{co:1} we obtain
\begin{align}\label{Eq:32}
f_2(p)\nonumber&\equiv\left\|\mathcal{M}_{\mu, \nu}^l \left(\rho_p\right)\right\|_{\mathrm{tr}} -\sqrt{\left(l \mu^{2}+\frac{a_{A}d_{A}^{2}+1}{d_{A}(d_{A}+1)}\right)\left(l \nu^{2}+\frac{a_{B}d_{B}^{2}+1}{d_{B}(d_{B}+1)}\right)} \\&\approx 0.011589p -0.0097194,
\end{align}
which implies that $\rho_p$ is entangled for $0.837 993 \leq p \leq 1$.
From Remark $2$ in Ref.\cite{tang2023enhancing},  we know that
\begin{align}\label{Eq:33}
g_2(p) \nonumber&\equiv \|\mathcal{G}(\rho_p)\|_{\mathrm{tr}} -\sqrt{\frac{a_{A}d_{A}^{2}+1}{d_{A}(d_{A}+1)}\frac{a_{B}d_{B}^{2}+1}{d_{B}(d_{B}+1)}}\\& \approx 0.010979p -0.0096898 > 0,
\end{align}
then $\rho_p$ is entangled, which gives rise to  $0.882 577 \leq p \leq 1$.
A direct comparison with the results of Ref.\cite{tang2023enhancing} in Fig.\ref{fig:7} demonstrates that our results better than that of Ref.\cite{tang2023enhancing}.
\begin{figure}
   \includegraphics[width=0.80\textwidth]{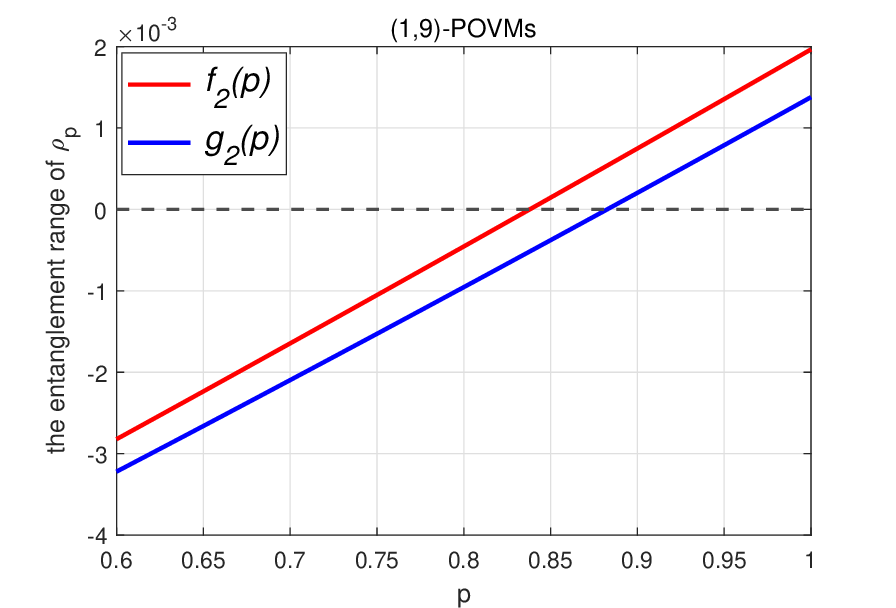}
    \caption{~$f_2(p)$~ derived from Corollary \ref{co:1} (solid red line), from which~ \(\rho_p\)~ is entangled for~$0.837 933 \leq p \leq 1$. ~$g_2(p)$~  presented in Remark $2$ of Ref.\cite{tang2023enhancing} (dashed blue line), from which~ \(\rho_p\)~ is entangled for~ $0.882 577 \leq p \leq 1$.}
    \label{fig:7}

\end{figure}

\begin{figure}
\includegraphics[width=0.80\textwidth]{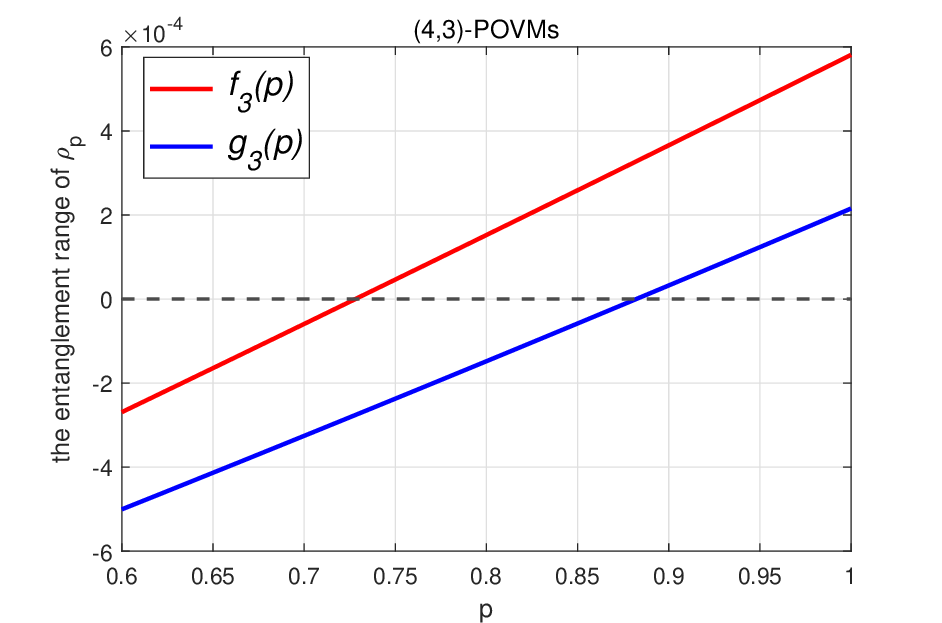}
    \caption{~\(f_3(p)\)~ derived from Corollary \ref{co:2} (solid red line), from which~ \(\rho_p\)~ is entangled for~ \(0.728 219 \leq p \leq 1\).  ~\(g_3(p)\)~ presented in Remark $3$ of Ref.\cite{tang2023enhancing} (dashed blue line), from which~ \(\rho_p\)~ is entangled for~ \(0.882 178 \leq p \leq 1\).}
    \label{fig:8}
\end{figure}

When \(N = 4\) and \(M = 3\),  Theorem \ref{th:1} reduces to Corollary \ref{co:2}, in which the Hermitian basis operator $G_{\alpha, k}$ is given by the following Gell-Mann matrices(see \cite{kalev2014mutually} for a detail):
\begin{align*}
G_{11} &= \frac{1}{\sqrt{2}}\left(
\begin{array}{ccc}
0 & 1 & 0 \\
1 & 0 & 0 \\
0 & 0 & 0
\end{array}
\right),~~
G_{12} = \frac{1}{\sqrt{2}}
\left(
\begin{array}{ccc}
0 & -i & 0 \\
i & 0 & 0 \\
0 & 0 & 0
\end{array}
\right),~~
G_{21} = \frac{1}{\sqrt{2}}
\left(
\begin{array}{ccc}
0 & 0 & 1 \\
0 & 0 & 0 \\
1 & 0 & 0
\end{array}
\right),~~
G_{22} = \frac{1}{\sqrt{2}}
\left(
\begin{array}{ccc}
0 & 0 & -i \\
0 & 0 & 0 \\
i & 0 & 0
\end{array}
\right),\\
G_{31} &= \frac{1}{\sqrt{2}}
\left(
\begin{array}{ccc}
0 & 0 & 0 \\
0 & 0 & 1 \\
0 & 1 & 0
\end{array}
\right),~~
G_{32} = \frac{1}{\sqrt{2}}
\left(
\begin{array}{ccc}
0 & 0 & 0 \\
0 & 0 & -i \\
0 & i & 0
\end{array}
\right),~~
G_{41} = \frac{1}{\sqrt{2}}
\left(
\begin{array}{ccc}
1 & 0 & 0 \\
0 & -1 & 0 \\
0 & 0 & 0
\end{array}
\right),~~
G_{42} = \frac{1}{\sqrt{6}}
\left(
\begin{array}{ccc}
1 & 0 & 0 \\
0 & 1 & 0 \\
0 & 0 & -2
\end{array}
\right).
\end{align*}
Here, the parameter is $\kappa=\frac{1}{3}+2t^{2}(1+\sqrt{3})^2$ with $t\in[-0.0547,0.3454]$.
In particular, for $\mu = 0.1$, $\upsilon =0.05$, $l = 2$ and $t = 0.01$, Corollary \ref{co:2} gives rise to
\begin{align}\label{Eq:34}
f_3(p) \nonumber&\equiv \left\|\mathcal{M}_{\mu, \nu}^l \left(\rho_p\right)\right\|_{\mathrm{tr}} - \sqrt{\left(l \mu^{2} + 1 + \kappa_A\right)\left(l \nu^{2}+ 1 + \kappa_B\right)} \\&\approx 0.002072p -0.0015086,
\end{align}
namely, $\rho_p$ is entangled for $0.728 219 \leq p \leq 1$. From Remark $3$ in Ref.\cite{tang2023enhancing}, we know that
\begin{align}\label{Eq:35}
g_3(p) \nonumber&\equiv \|\mathcal{J}(\rho_p)\|_{\mathrm{tr}} -\sqrt{(1 + \kappa_A)(1 + \kappa_B)} \\&\approx 0.001718p -0.00151559 >0,
\end{align}
then $\rho_p$ is entangled; i.e., $0.882 178 \leq p \leq 1$. Clearly, our result outperforms that of Ref.\cite{tang2023enhancing}; see Fig.\ref{fig:8}.

\end{example}

\begin{example}\label{Ex:4}
We consider one distinct class of PPT states $\rho_1(\lambda)$ given by \cite{bandyopadhyay2005non,halder2019construction}:
\begin{align}
\rho_1(\lambda) &= \lambda |\omega_1\rangle \langle \omega_1| + (1 - \lambda)\rho_{BE},
\end{align}
where
\begin{align}
\rho_{BE} &= \frac{1}{4}\left( I_{9} - \sum_{i=1}^{5}|\omega_i\rangle\langle\omega_i|\right)
\end{align}
is a bound entangled state with
\begin{align}
\nonumber|\omega_1\rangle& = |2\rangle \otimes \frac{1}{\sqrt{2}}(|1\rangle - |2\rangle),\qquad
\nonumber|\omega_2\rangle = |0\rangle \otimes \frac{1}{\sqrt{2}}(|0\rangle - |1\rangle),\\
|\omega_3\rangle &= \frac{1}{\sqrt{2}}(|0\rangle - |1\rangle) \otimes |2\rangle,\qquad
\nonumber|\omega_4\rangle = \frac{1}{\sqrt{2}}(|1\rangle - |2\rangle) \otimes |0\rangle,\\
\left| \omega_5 \right\rangle &= \frac{1}{\sqrt{3}} (|0\rangle + |1\rangle + |2\rangle) \otimes \frac{1}{\sqrt{3}} (|0\rangle + |1\rangle + |2\rangle).
\end{align}
\end{example}
Clearly, $\rho_1(\lambda)$ is of rank $5$. Set $\mu = \nu = 0.005$, $t = 0.01$ and $l = 1$. In Table \ref{tab:3} we compare our results with the existing ones. It can be seen that our criterion detects entangled states better than those in  Refs.\cite{chen2014entanglement,chen2015general,{tang2023improved},
{shen2018improved},{lai2018entanglement},{Siudzinska2022}}.

\begin{center}
	\begin{table*}[htp!]
\caption{Entanglement of the state $\rho_1(\lambda)$ in Example \ref{Ex:4}}
		\label{tab:3}
		\begin{tabular}{cccc}
			\hline\noalign{\smallskip}
		 \quad& (4,3) POVMs\quad&(1,9)  POVMs\quad & (8,2) POVMs\\
		\noalign{\smallskip}\hline\noalign{\smallskip}
		Tensor product based \quad& fail to detect \cite{chen2014entanglement} \quad &  fail to detect  \cite{chen2015general}  \quad&  fail to detect \cite{tang2023improved}  \\	
		\hspace{1mm}\\
Trace norm based \quad&$0 \leq \lambda \leq 0.069160$ \cite{shen2018improved} \quad& $0 \leq \lambda \leq 0.068848$ \cite{lai2018entanglement} \quad& $0 \leq \lambda \leq 0.069163$ \cite{Siudzinska2022} \\		
		\hspace{1mm}\\
Our method \quad& $0 \leq \lambda \leq 0.070740$ \quad& $0 \leq \lambda \leq 0.069089$ \quad& $0 \leq \lambda \leq 0.072155$ \\
		\noalign{\smallskip}\hline
		\end{tabular}
	\end{table*}
\end{center}

\begin{example}\label{Ex:2}
Consider the isotropic states
\begin{eqnarray}\label{Eq:36}
\rho_{\mathrm{iso}}=q\mid\Phi^{+}\rangle\langle\Phi^{+}|+(1-q)\frac{\mathbb{I}}{d^2},\ \ \  0\leq q\leq1,
\end{eqnarray}
where $\mid\Phi^{+}\rangle=\displaystyle\frac{1}{\sqrt{d}}\sum\limits_{i=0}\limits^{d-1}\mid i\rangle\mid i\rangle$. We set $u = v = 2$, $l = 10$ and $t = 0.01$ within Theorem \ref{th:1}, Corollary \ref{co:1} and Corollary \ref{co:2}. Then, we have the following conclusions for $d=3$:

{\it (8,2) POVMs}~~From (\ref{eq:5}) we obtain

\begin{align}\label{Eq:37}
f_4(q)\nonumber&\equiv\left\|\mathcal{M}_{\mu, \nu}^l \left(\rho_{\mathrm{iso}}\right)\right\|_{\mathrm{tr}} - \sqrt{\left(l \mu^{2}+\frac{(d_{A}-1)(d_{A}^2+M_{A}^2x_A)}{d_{A}M_{A}(M_{A}-1)}\right)\left(l \nu^{2}+\frac{(d_{B}-1)(d_{B}^2+M_{B}^2x_B)}{d_{B}M_{B}(M_{B}-1)}\right)} \\&  \approx 0.0032q-0.0008.
\end{align}

{\it(1,9) POVMs} The $(N, M)$-POVMs reduce to GSICPOVMs. By using (\ref{eq:6}), we obtain

\begin{eqnarray}\label{Eq:38}
f_5(q)\equiv\left\|\mathcal{M}_{\mu, \nu}^l \left(\rho_{\mathrm{iso}}\right)\right\|_{\mathrm{tr}} -\sqrt{\left(l \mu^{2}+\frac{a_{A}d_{A}^{2}+1}{d_{A}(d_{A}+1)}\right)\left(l \nu^{2}+\frac{a_{B}d_{B}^{2}+1}{d_{B}(d_{B}+1)}\right)} \approx 0.0384q-0.0096.
\end{eqnarray}

{\it(4,3) POVMs} The $(N, M)$-POVMs reduce to GSICPOVMs.  From (\ref{eq:8}) we have

\begin{eqnarray}\label{Eq:39}
f_6(q) \equiv \left\|\mathcal{M}_{\mu, \nu}^l \left(\rho_{\mathrm{iso}}\right)\right\|_{\mathrm{tr}} - \sqrt{\left(l \mu^{2} + 1 + \kappa_A\right)\left(l \nu^{2}+ 1 + \kappa_B\right)} \approx 0.0060q-0.0015.
\end{eqnarray}

The above calculations show that $\rho_{\mathrm{iso}}$ is entangled for $\frac{1}{4} < q \leq 1$. Therefore, our criterion detects all the entanglement of the isotropic states due to the fact that $\rho_{\mathrm{iso}}$ is entangled only if $q > \frac{1}{d+1}$ \cite{bertlmann2005optimal}, as shown in Fig.\ref{fig:3}.
\begin{figure}[htp]
\includegraphics[width=0.80\textwidth]{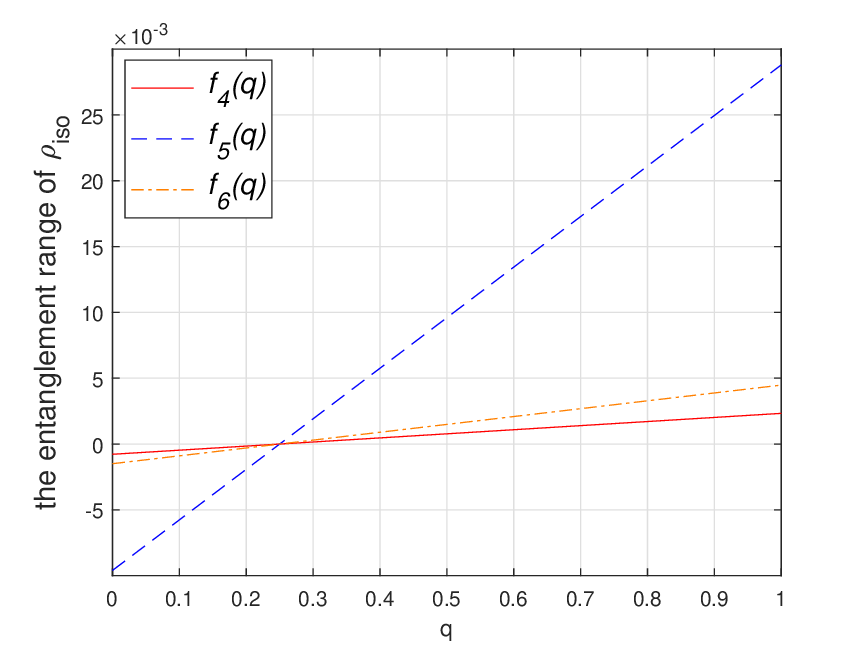}
\caption{~\(f_4(q)\)~ in Theorem \ref{th:1} (solid red line), ~\(f_5(q)\)~ in Corollary \ref{co:1} (dashed blue line), and ~\(f_6(q)\)~ in Corollary \ref{co:2} (dash-dotted orange line). It can be seen  that ~\(\rho_{\mathrm{iso}}\)~ is entangled for ~\(\frac{1}{4} < q \leq 1\).}
\label{fig:3}
\end{figure}

\end{example}

The matrix realignment has played an important role in the studies of separability  \cite{chen2002matrix,zhang2008entanglement,shen2015separability}. In Ref.\cite{shi2023family} Shi $et$ $al.$ defined the following matrix:

\begin{eqnarray}\label{Eq:40}
\mathcal{Q}_{\alpha,\beta}(\rho_{AB})=\begin{pmatrix}
	\alpha\beta& \alpha \mathrm{Vec}(\rho_B)^T\\
	\beta \mathrm{Vec}(\rho_A)& \mathrm{R}(\rho_{AB})
\end{pmatrix},
\end{eqnarray}
where $\alpha,\beta\in \mathbb{R}$, $\rho_A$ and $\rho_B$ are the reduced density matrices of systems $A$ and $B$, respectively; $\mathrm{R}(\rho_{AB})$ is the realigned matrix of $\rho_{AB}$; $\mathrm{Vec}(X)$ denotes the vectorization of the matrix $X$. It is shown that if $\rho_{AB}$ is separable, then
$\norm{\mathcal{Q}_{\alpha,\beta}(\rho_{AB})}_{\mathrm{tr}}\le \sqrt{(\alpha^2+1)(\beta^2+1)}$.
In Ref.\cite{sun2025separability} Sun $et$ $al.$ introduced the following matrix:

\begin{eqnarray}\label{Eq:41}
	\mathcal{Q}_{\alpha, \beta}^{l}(\rho_{AB})=\left(\begin{array}{cc}
		\alpha \beta J_{l \times l} & \alpha \omega_{l}\left(\mathrm{tr}_{A}(\rho)\right)^{T} \\
		\beta \omega_{l}\left(\mathrm{tr}_{B}(\rho)\right) & \mathrm{R}(\rho_{AB})
	\end{array}\right),
\end{eqnarray}
where $\alpha$ and $\beta$ are arbitrary real numbers, $l$ is a natural number, $J_{l \times l}$ is the matrix with all $l \times l$ elements being $1$, and $\mathrm{tr}_{A}$ denotes the partial trace over the subsystem $A$. It is shown that if $\rho_{AB} \in H_{A} \otimes H_{B}$ is separable, then $\left\|\mathcal{Q}_{\alpha, \beta}^{l}\left(\rho_{AB}\right)\right\|_{\mathrm{tr}} \leq \sqrt{\left(l \alpha^{2}+1\right)\left(l \beta^{2}+1\right)}$. We give an example to compare our results with the ones given in Refs.\cite{shi2023family,sun2025separability}.

\begin{example}\label{Ex:3}
Consider the mixture of the bound entangled state \cite{horodecki1997separability}
	\begin{equation}\label{Eq:42}
		\rho_{\upsilon}=\dfrac{1}{1+8\upsilon}
		\begin{pmatrix}
			\upsilon & 0 & 0 & 0 & \upsilon & 0 & 0 & 0 & \upsilon\\
			0 & \upsilon & 0 & 0 & 0 & 0 & 0 & 0 & 0\\
			0 & 0 & \upsilon & 0 & 0 & 0 & 0 & 0 & 0\\
			0 & 0 & 0 & \upsilon & 0 & 0 & 0 & 0 & 0\\
			\upsilon & 0 & 0 & 0 & \upsilon & 0 & 0 & 0 & \upsilon\\
			0 & 0 & 0 & 0 & 0 & \upsilon & 0 & 0 & 0\\
			0 & 0 & 0 & 0 & 0 & 0 & \frac{1+\upsilon}{2} & 0 & \frac{\sqrt{1-\upsilon^{2}}}{2}\\
			0 & 0 & 0 & 0 & 0 & 0 & 0 & \upsilon & 0\\
			\upsilon & 0 & 0 & 0 & \upsilon & 0 & \frac{\sqrt{1-\upsilon^{2}}}{2} & 0 & \frac{1+\upsilon}{2}
		\end{pmatrix}
	\end{equation}
and the identity,
    \begin{equation}\label{Eq:43}
    	\rho_y=y\rho_{\upsilon}+\frac{1-y}{9}I_{9}.
    \end{equation}
We set $N=8$ and $M=2$. Table \ref{tab:2} presents a comparison of the results derived from our Theorem \ref{th:1}
with $\mu =\nu=2$ and $l=10$, with  those of Theorem $1$ of Ref.\cite{shi2023family}  and Theorem $1$ of Ref.\cite{sun2025separability} for various values of $\upsilon$.
\begin{center}
	\begin{table*}[htp!]
\caption{Entanglement of the state $\rho_{y}$ in Example \ref{Ex:3} for different values of $\upsilon$}
		\label{tab:2}
		\begin{tabular}{cccc}
			\hline\noalign{\smallskip}
		$\upsilon$ \quad& Theorem $1$ in Ref.\cite{shi2023family}\quad&Theorem $1$ in Ref.\cite{sun2025separability} 
\quad & Our Theorem \ref{th:1} (t=0.01)\\
		\noalign{\smallskip}\hline\noalign{\smallskip}
		0.2 \quad& $0.9943 \leq y \leq 1$ \quad & $0.99408  \leq y \leq 1$ \quad& $0.994054  \leq y \leq 1$ \\
		
		0.4 \quad& $0.9948  \leq y \leq 1$ \quad& $0.99463  \leq y \leq 1$ \quad& $0.994609  \leq y \leq 1$ \\
		
		0.6 \quad& $0.9964  \leq y \leq 1$ \quad& $0.99627  \leq y \leq 1$ \quad& $0.99625 \leq y \leq 1$ \\
		
		0.8 \quad& $0.9982  \leq y \leq 1$ \quad& $0.998127  \leq y \leq 1$ \quad& $0.998122  \leq y \leq 1$ \\
		
		0.9 \quad& $0.9991  \leq y \leq 1$ \quad& $0.99906892  \leq y \leq 1$ \quad& $0.9990664  \leq y \leq 1$ \\
		\noalign{\smallskip}\hline
		\end{tabular}
	\end{table*}
\end{center}

 For $\mu = \nu = 2,~l = 10,~t=0.01$ and $\upsilon=0.2$, Theorem \ref{th:1} shows that if $f_7(y)\equiv \left\|\mathcal{M}_{\mu, \nu}^l \left(\rho_y\right)\right\|_{\mathrm{tr}} - \sqrt{\left(l \mu^{2}+\frac{(d_{A}-1)(d_{A}^2+M_{A}^2x_A)}{d_{A}M_{A}(M_{A}-1)}\right)\left(l \nu^{2}+\frac{(d_{B}-1)(d_{B}^2+M_{B}^2x_B)}{d_{B}M_{B}(M_{B}-1)}\right)} >0$, $\rho_y$ is entangled for $0.994054  \leq y \leq 1$.  Theorem $1$ in Ref.\cite{sun2025separability} says that if $h_1(y)\equiv\left\|\mathcal{Q}_{\alpha, \beta}^{l}\left(\rho_y\right)\right\|_{\mathrm{tr}} - \sqrt{\left(l \alpha^{2}+1\right)\left(l \beta^{2}+1\right)} >0$, $\rho_y$ is entangled, that is, $0.99408  \leq y \leq 1$.  Theorem $1$ in Ref.\cite{shi2023family} implies that $\rho_y$ is entangled if $h_2(y)\equiv\norm{\mathcal{Q}_{\alpha,\beta}(\rho_y)}_{\mathrm{tr}} - \sqrt{(\alpha^2+1)(\beta^2+1)} >0$, i.e., $0.9943 \leq y \leq 1$. It can be clearly seen from Table \ref{tab:2} and Fig.\ref{fig:2}  that compared with the methods in Refs.\cite{sun2025separability,shi2023family}, our method is capable of identifying a much wider range of entangled states.
\begin{figure}[htp]
\includegraphics[width=0.80\textwidth]{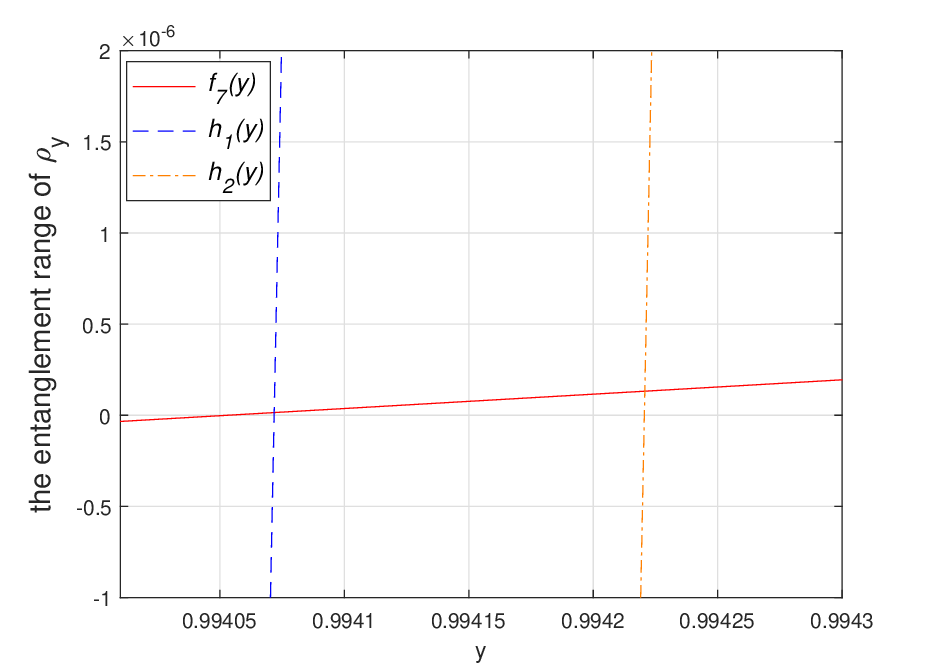}
\caption{ ~\(f_7(y)\)~ derived from Theorem \ref{th:1} (solid red line), from which ~\(\rho_y\)~ is entangled for ~\(0.994054 \leq y \leq 1\). \(h_1(y)\)~ presented in Theorem $1$ of Ref.\cite{sun2025separability} (dashed blue line),~from which~\(\rho_y\)~ is entangled for ~\(0.99408 \leq y \leq 1\).~\(h_2(y)\)~ given in Theorem $1$ of Ref.\cite{shi2023family} (dash-dotted orange line), from which ~\(\rho_y\)~ is entangled for~\(0.9943 \leq y \leq 1\).}
\label{fig:2}
\end{figure}
\end{example}

\section{DETECTION OF MULTIPARTITE ENTANGLEMENT}\label{sec:4}

We generalize  Theorem \ref{th:1} to tripartite and multipartite quantum systems. Let $\rho_{ABC}\in H_{A} \otimes H_{B} \otimes H_{C}$ be a tripartite state, and  $\{E_{\alpha,k}\}_{\alpha = 1}^{N_{A}}$, $\{E_{\beta,j}\}_{\beta = 1}^{N_{B}}$ and $\{E_{\gamma,n}\}_{\gamma = 1}^{N_{C}}$ the $(N,M)$ POVMs on $H_A$, $H_B$ and $H_C$ with efficiency parameter $x_{A}$, $x_{B}$ and $x_{C}$, respectively. We have the following theorem.

\begin{theorem}\label{th:2}
If \(\rho_{ABC} = \sum_i p_i \rho_i^A \otimes \rho_i^B \otimes \rho_i^C\) is fully separable, the following inequalities must be simultaneously satisfied:

(1) Inequality for the bipartition \(A|BC\):
\begin{eqnarray}\label{Eq:44}
\left\| \mathscr{M}_{\mu,\nu}^l(\rho_{A|BC}) \right\|_{\text{tr}} \leq \sqrt{\left( l\mu^2 + \frac{(d_A - 1)(d_A^2 + M_A^2 x_A)}{d_A M_A (M_A - 1)} \right) \left( l\nu^2 + \frac{(d_B - 1)(d_B^2 + M_B^2 x_B)}{d_B M_B (M_B - 1)} \cdot \frac{(d_C - 1)(d_C^2 + M_C^2 x_C)}{d_C M_C (M_C - 1)} \right)},
\end{eqnarray}

(2) Inequality for the bipartition \(B|AC\):
\begin{eqnarray}\label{Eq:45}
\left\| \mathscr{M}_{\mu,\nu}^l(\rho_{B|AC}) \right\|_{\text{tr}} \leq \sqrt{\left( l\mu^2 + \frac{(d_B - 1)(d_B^2 + M_B^2 x_B)}{d_B M_B (M_B - 1)} \right) \left( l\nu^2 + \frac{(d_A - 1)(d_A^2 + M_A^2 x_A)}{d_A M_A (M_A - 1)} \cdot \frac{(d_C - 1)(d_C^2 + M_C^2 x_C)}{d_C M_C (M_C - 1)} \right)},
\end{eqnarray}

(3) Inequality for the bipartition \(C|AB\):
\begin{eqnarray}\label{Eq:46}
\left\| \mathscr{M}_{\mu,\nu}^l(\rho_{C|AB}) \right\|_{\text{tr}} \leq \sqrt{\left( l\mu^2 + \frac{(d_C - 1)(d_C^2 + M_C^2 x_C)}{d_C M_C (M_C - 1)} \right) \left( l\nu^2 + \frac{(d_A - 1)(d_A^2 + M_A^2 x_A)}{d_A M_A (M_A - 1)} \cdot \frac{(d_B - 1)(d_B^2 + M_B^2 x_B)}{d_B M_B (M_B - 1)} \right)}.
\end{eqnarray}

Conversely, if \(\rho_{ABC}\) violates any one  of the above inequalities, \(\rho_{ABC}\) is an entangled state.

\end{theorem}

\begin{proof}
To begin, we first prove that for any bipartite separable state $\rho_{AB} \in H_A \otimes H_B$,
\begin{align}\label{eq:20}
 &\nonumber\sum_{\alpha,\beta = 1}^{N} \sum_{k,j = 1}^{M} \mathrm{tr}(E_{\alpha,k} \otimes E_{\beta,j}\rho_{AB})^2\\ &\nonumber= \sum_{\alpha,\beta = 1}^{N} \sum_{k,j = 1}^{M} \mathrm{tr}(E_{\alpha,k}\rho_{A})^2 \mathrm{tr}(E_{\beta,j}\rho_{B})^2\\ &\nonumber= \sum_{\alpha}^{N} \sum_{k = 1}^{M} \mathrm{tr}(E_{\alpha,k}\rho_{A})^2 \sum_{\beta}^{N} \sum_{j = 1}^{M}\mathrm{tr}(E_{\beta,j}\rho_{B})^2 \\ &\leq(\frac{d_{A} - 1}{d_{A}} \frac{d_{A}^{2}+M_{A}^{2}x_{A}}{M_{A}(M_{A}- 1)}) (\frac{d_{B}-1}{d_{B}} \frac{d_{B}^{2}+M_{B}^{2}x_{B}}{M_{B}(M_{B}-1)}).
\end{align}
Therefore, we have
\begin{align}\label{Eq:48}
  \left\|\mathcal{M_{\mathrm{\underline{A}BC}}}_{\mu, \nu}^l \left(\rho_{ABC}\right)\right\|_{t r} \nonumber&=\left\|\left(\begin{array}{cc}\mu\nu J_{l\times l} & \mu \omega_{l}(\sigma_{BC})^T \\ \nu \omega_{l}(\tau_A)  \nonumber& \tau_A \sigma_{BC}^\mathrm{T}\end{array}\right)\right\|_{\mathrm{tr}} \\
\nonumber&=\left\|\left(\begin{array}{cc}\mu J_{l\times 1} & \\ \tau_A & \end{array}\right)\left(\begin{array}{cc}\nu J_{1\times l} & \sigma_{BC}^\mathrm{T}\end{array}\right)\right\|_{\mathrm{tr}} \\
&\leq\left\|\left(\begin{array}{cc}\mu J_{l\times 1} & \\ \tau_A  & \end{array}\right)\right\|_{\mathrm{tr}} \left\|\left(\begin{array}{cc}\nu J_{1\times l} & \sigma_{BC}^\mathrm{T}\end{array}\right)\right\|_{\mathrm{tr}} \\
\nonumber&=\left\|\left(\begin{array}{cc}\mu J_{l\times 1} & \\ \tau_A  & \end{array}\right)\right\|_{\mathrm{tr}} \left\|\left(\begin{array}{cc}\nu J_{1\times l} & \sigma_{BC}^\mathrm{T}\end{array}\right)\right\|_{\mathrm{tr}} \\
\nonumber&\leq \sqrt{\left(l \mu^{2}+\frac{(d_{A}-1)(d_{A}^2+M_{A}^2x_A)}{d_{A}M_{A}(M_{A}-1)}\right)\left(l \nu^{2}+\frac{(d_{B}-1)(d_{B}^2+M_{B}^2x_B)}
{d_{B}M_{B}(M_{B}-1)}\frac{(d_{C}-1)(d_{C}^2+M_{C}^2x_C)}{d_{C}M_{C}(M_{C}-1)}\right)},
\end{align}
where the last inequality is due to ({\ref{re:1}}) and ({\ref{eq:20}}).
The last two inequalities  can be proved in a similar way.
\end{proof}


For an $N$-partite state $\rho_{A_1 A_2 \ldots A_{\mathcal{N}}}$ in $H_{A_1} \otimes H_{A_2} \otimes \ldots \otimes H_{A_{{N}}}$, we have  the a similar conclusion.

\begin{theorem}\label{th:4}
 If a quantum state
$\rho_{A_1 A_2 \cdots A_{\mathcal{N}}}\in H_{A_1} \otimes H_{A_2} \otimes \ldots \otimes H_{A_{{N}}}$
 violates the following inequality, it must be entangled,
\begin{eqnarray}\label{eq:15}
  \left\|\mathcal{M_{\mathrm{\underline{A_{q}}A_{1}\ldots A_{q-1}A_{q+1}\ldots A_{{N}}}}}_{\mu, \nu}^l \left(\rho_{{A_1 A_2 \ldots A_{\mathcal{N}}}}\right)\right\|_{\mathrm{tr}} &\leq \sqrt{\left(l \mu^{2}+\frac{(d_{A_q}-1)(d_{A_q}^2+M_{A_q}^2x_A)}{d_{A_q}M_{A_q}(M_{A_q}-1)}\right)\left(l \nu^{2}+\prod_{i=1,i \neq q }^{{N}}\frac{d_{A_i} - 1}{d_{A_i}} \frac{d_{A_i}^2 + M_{A_i}^2 x_{A_i}}{M_{A_i} (M_{A_i} - 1)}\right)}.
\end{eqnarray}
\end{theorem}

\section{CONCLUSIONS}\label{sec:5}
We have presented separability criteria based on $(N,M)$ POVMs. Through detailed  examples we have demonstrated that our approach exhibits significant advantages in entanglement detection, compared with some existing criteria. Different from other similar approaches,  we used the $(N,M)$ POVMs measurement probabilities in Eq.(\ref{eq:3}) instead of the reduced density matrices. Our criteria allow for detecting quantum entanglement experimentally without tomography. Our results may be applied to derive new lower bounds of concurrence and inspire further investigations on the theory of quantum entanglement.

Building upon these advantages, several promising research directions emerge. Firstly, the $(N,M)$ POVMs have been extended to the generalized symmetric measurements and generalized equiangular measurements in Refs.\cite{siudzinska2024informationally,siudzinska2024much}.
It would be interesting to apply quantum measurements such as these to investigate separability in bipartite or multipartite systems \cite{KS}. Secondly, following our approach, specific matrices based on the corresponding measurement probabilities could be constructed to improve entanglement detection capability. Furthermore, improved separability criteria may lead to better lower bounds for concurrence. Thirdly, in Ref.\cite{qi2025k}, the authors analyzed $k$-partite entanglement based on the maximum unbiased measurements. Our method could be further reasonably extended to the study of $k$-partite entanglement and genuine multipartite entanglement, by a similar matrix construction of probabilities from the measurements independently on two of the subsystems.


\bigskip
\noindent{\bf Acknowledgements}
~This work is supported by the National Natural Science Foundation of China (NSFC) under Grant No. 12171044, the specific research fund of the Innovation Platform for Academicians of Hainan Province.

\end{document}